\documentclass{article}
\usepackage[margin=1in]{geometry}
\usepackage[colorlinks=true, allcolors=blue]{hyperref}
\usepackage{authblk}
\usepackage[utf8x]{inputenc}
\usepackage{graphicx}
\usepackage{multirow,array}
\usepackage{amsmath}
\usepackage{amssymb}
\usepackage{amsthm}
\usepackage{float}
\usepackage{textcomp}
\usepackage{url}
\usepackage{ulem}
\usepackage{xcolor}
\usepackage{wasysym} 
\usepackage{braket}
\usepackage{pst-node}
\usepackage{tikz} 
\usepackage{pgf}
\usepackage{tikz-cd} 
\usepackage{dsfont}
\usepackage{epigraph}
\usepackage{blindtext}
\usepackage{footnote}
\usepackage{lipsum}
\usepackage{microtype}
\usepackage{ragged2e}
\usepackage{ulem}
\usepackage{enumerate}
\usepackage{enumitem}
\usepackage{mathtools}
\usepackage[new]{old-arrows}
\usepackage{physics}
\usepackage{setspace}
\usepackage{color}
\usepackage[all,cmtip]{xy}
\usepackage{ucs}
\usepackage{soul}
\usepackage{pst-node}
\usepackage{dutchcal}

\hypersetup{
	colorlinks=true,
	linkcolor=blue,
	filecolor=blue,
	citecolor = blue,
	urlcolor=blue      
}
\usetikzlibrary{matrix,arrows,decorations.pathmorphing,quotes, automata, positioning, arrows}
\definecolor{Blue}{rgb}{0,0,0.9}
\definecolor{Red}{rgb}{0.9,0,0}

\newtheorem{definition}{Definition}[section]
\newtheorem*{definition*}{Definition}

\newtheorem*{teorema*}{Teorema}

\newtheorem*{demo*}{Demonstração}

\newtheorem*{lema*}{Lema}
\newtheorem{prop}{Proposition}[section]
\newtheorem*{prop*}{Proposition}
\newtheorem{corolario}{Corollary}[section]
\newtheorem*{corolario*}{Corollary}

\newtheorem*{afirm*}{Afirmação}

\newtheorem*{axiom*}{Axioma}
\newtheorem{premise}{Premise}[section]
\newtheorem*{premise*}{Premise}
\newtheorem*{premise*3.2}{Premise 3.2.2}
\newtheorem*{premise*3.1}{Premise 3.2.1}
\theoremstyle{remark}
\newtheorem{remark}{\normalfont \textbf{Remark}}[section]
\newtheorem*{remark*}{\normalfont \textbf{Remark}}

\definecolor{quotemark}{gray}{0.7}
\makeatletter
\def\fquote{%
	\@ifnextchar[{\fquote@i}{\fquote@i[]}
}%
\def\fquote@i[#1]{%
	\def\tempa{#1}%
	\@ifnextchar[{\fquote@ii}{\fquote@ii[]}
}%
\def\fquote@ii[#1]{%
	\def\tempb{#1}%
	\@ifnextchar[{\fquote@iii}{\fquote@iii[]}
}%
\def\fquote@iii[#1]{%
	\def\tempc{#1}%
	\vspace{1em}%
	\noindent%
	\begin{list}{}{%
			\setlength{\leftmargin}{0.1\textwidth}%
			\setlength{\rightmargin}{0.1\textwidth}%
		}%
		\item[]%
		\begin{picture}(0,0)%
			\put(-15,0){\makebox(0,0){\scalebox{3}{\textcolor{quotemark}{``}}}}%
		\end{picture}%
		\begingroup\itshape}%
	
	\def\endfquote{%
		\endgroup\par%
		\makebox[0pt][l]{%
			\hspace{0.8\textwidth}%
			\begin{picture}(0,0)(0,0)%
				\put(-246,15){\makebox(0,0){%
						\scalebox{3}{\color{quotemark}''}}}%
		\end{picture}}%
		\ifx\tempa\empty%
		\else%
		\ifx\tempc\empty%
		\hfill\rule{100pt}{0.5pt}\\\mbox{}\hfill\tempa,\ \emph{\tempb}%
		\else%
		\hfill\rule{100pt}{0.5pt}\\\mbox{}\hfill\tempa,\ \emph{\tempb},\ \tempc%
		\fi\fi\par%
		\vspace{0.5em}%
	\end{list}%
}%
\makeatother

\begin{document}
	
	\title{A Categorical View on the Principle of Relativity}
	
	\author[1]{L. M. Gaio}
	\author[2]{B. F. Rizzuti}
	
	\affil[1]{Programa de P\'os-Gradua\c{c}\~ao em Matem\'atica, Instituto de Ci\^encias Exatas,  Universidade Federal de Juiz de Fora 36036-900, Juiz de Fora - MG, Brazil\\  
 
 lucamauadgaio@ice.ufjf.br} 
	\affil[2]{Departamento de F\'isica, Instituto de Ci\^encias Exatas, Universidade Federal de Juiz de Fora, 36036-900, Juiz de Fora - MG, Brazil\\ 
 
 brunorizzuti@ice.ufjf.br}

	\date{}                     
	\setcounter{Maxaffil}{0}
	\renewcommand\Affilfont{\itshape\small}

	
	

	
	\maketitle  
	
	\begin{abstract}
		Category theory plays a special character in mathematics - it unifies distinct branches under the same formalism. Despite this integrative power in math, it also seems to provide the proper foundations to the experimental physicist. In this work, we present another application of category in physics, related to the principle of relativity. The operational construction of (inertial) frames of reference indicates that only the movement between one and another frame is enough to differentiate both of them. This fact is hidden when one applies only group theory to connect frames. In fact, rotations and translations only change coordinates, keeping the frame inert. The change of frames is only attainable by boosts in the classical and relativistic regimes for both Galileo and Lorentz (Poincaré) groups. Besides providing a non-trivial example of application of category theory in physics, we also fulfill the presented gap when one directly applies group theory for connecting frames.
	\end{abstract}
	
	\textbf{Keywords:} Category theory, Principle of relativity, Operationalism.

\section{Introduction}
	\label{sec.1}
	
Category theory carries a special character in mathematics - it unifies many distinct areas which, at first glance, seem to have no resemblance. For instance, the lowest common multiple and direct sum of vector spaces as well as discrete topologies, fields of fractions and free groups are standard examples of structures assembled under the same umbrella \cite{leinster_basic_2014}. Despite this unifying power of mathematical constructions, we could look at category theory applied to physics. It is argued that it is precisely the structure the experimental physicist needs \cite{coecke_categories_2009}. We may be convinced of this claim with the following line of reasoning so described in \cite{bob_2006}. Let a physical system be denoted by $A$, which may represent a qubit, an electron, or even classical measurement data. The experimentalist performs an operation $f$ on it, resulting in a system $B$. We can think of $f$ as a measurement on the initial system and $B$ is a proxy to the resulting system together with the data encoding the measurement outcome. Schematically, this process is described by 
	\begin{equation}
		A \stackrel{f}{\longrightarrow} B.
	\end{equation}
	
	Other operations are allowed, such as 
	\begin{equation}
		B \stackrel{g}{\longrightarrow} C.
	\end{equation}
	Then, it is natural to expect a composition $g \circ f$ for the concatenation of the two operations. If we think of the composition as just one operation, we could also ask that
	\begin{equation}
		(h \circ g) \circ f = h \circ (g \circ f). 
	\end{equation}
	Well, to do nothing on $A$ can be accommodated by an identity 
	\begin{equation}
		A \stackrel{1_A}{\longrightarrow} A
	\end{equation}
	and, as a consequence, 
	\begin{equation}
		1_B \circ f = f \circ 1_A = f.
	\end{equation}
	
	This is exactly the structure of a category, that we now formalize with the following \cite{adamek2009abstract}.
	\begin{definition}
		A category $\mathcal{A}$ is composed of
		\begin{enumerate}[label=(\alph*)]
			\item a class $\mathcal{O}$ whose elements are called the objects of the category;
			
			\item for each pair of objects $A, B \in \mathcal{O}$ a class of maps (or morphisms) $f: A \rightarrow B$, denoted by $\mathcal{A}(A,B)$ which are pairwise disjoint and such that
			\begin{enumerate}[label=(\roman*)]
				\item for objects $A, B, C \in \mathcal{O}$ an associative composition is given
				\begin{align*}
					\mathcal{A}(A,B) \times \mathcal{A}(B,C) &\rightarrow \mathcal{A}(A,C) \\
					(f,g) &\mapsto g \circ f
				\end{align*}
				
				\item For each object $A \in \mathcal{O}$, there exists the identity map $\mathrm{id}_A : A \rightarrow A$, such that for each $f \in \mathcal{A}(A,B)$, $\mathrm{id}_B \circ f = f \circ \mathrm{id}_A = f$. 
			\end{enumerate}
		\end{enumerate}
	\end{definition}
	
	One of its applications is known as a Resource Theory, which is basically a category with a binary operation that allows one to compose objects \cite{coecke_mathematical_2016}. There is a wide range of applications, that runs from a resource theory of work and heat \cite{sparaciari_resource_2017} to quantum entanglement, passing through quantum computation \cite{chitambar_quantum_2019}. We may also trace some recent applications of category theory concerning foundational structure of classical field theories (e.g., Newtonian gravitation, general relativity and Yang-Mills theories) \cite{weatherall_categories_2015}. At last, we could cite the ref. \cite{baez_lauda_2011}, that exhibits a chronology of the growing role of categories and ($n$-categories) in areas such as Feynman diagrams, spin networks, string theory, loop quantum gravity, and topological quantum field theory. 
	
We go a step further in the direction of the deep liaison between category theory and physics, particularly connected with relativity. Actually, there is an entire research program focused on the ``categorification" of both special and general relativity theories. For instance, in \cite{guts}, the authors explore various approaches rooted in topos theory to construct the general relativity theory. Their perspective arises from the realization that when we view spacetime as a mere set of events, we cannot accommodate more intricate structures, such as time loops. Therefore, they advocate for the inclusion of topoi in the framework, which allows for a more complex structure of causal interactions between events. Another interesting proposal in this direction, from a more pragmatic standpoint, can be found in the reference \cite{Oziewicz}. The author introduces the concept of categorical relativity by considering a groupoid category of massive bodies in mutual motion. In this framework, the objects correspond to the massive bodies (which need not be inertial), while the morphisms represent the relative velocities between them. This stands in contrast to the reciprocal-velocity concept defined using Lorentz boosts in isometric special relativity.

A work that aligns closely with the topic being presented here is \cite{carvalho}. Following the approach of examining relativity theory through a categorical lens, the author defines specific categories in which the objects are inertial frames as well as categories where the objects are coordinate systems. In the former, the morphisms are boosts, while in the latter, Galileo (and Lorentz) transformations between coordinates serve as the corresponding morphisms. Inertial frames are connected to coordinates through functors, encompassing both the classical and relativistic domains.

The overarching goal of this approach is to address a fundamental assumption pertaining to the universal nature of physics. Specifically, the notion of being observer-free rather than solely Lorentz or diffeomorphic-invariant. In keeping with this methodology, in this work we will show how the homogeneity of spacetime can be extracted from pure categorical terms. As consequence, if no particular region is privileged, we can infer the principle of relativity, which asserts that the laws of physics must be covariant under the suitable transformations between inertial frames of reference. It may sound at first, that all this subject has already been run out via group theory. However, only a change such as \cite{naber_geometry_2012}
	\begin{equation}
		x'= \Lambda x
	\end{equation}
	does not differentiate a true change of reference from a simple change of coordinates. Here $x$ (and $x'$) denotes coordinates in Minkowski spacetime and $\Lambda$ is a Lorentz transformation. Actually, they denote the standard representation of the Lorentz group on the vector space formed by pair of events. In fact, $\Lambda$ encodes both rotations and boosts. The former does not change the frame of reference, while the latter does. The same argument also holds for translations, if we consider the Poincaré group, $x' = x + a$ - the reference is kept fixed, while changing only coordinates. Hence, we can conclude that change of references is only attainable when the two are connected by a boost. This inference is even clearer when we adopt an operational approach to construct underlying notions, such as the central one of a frame of reference \cite{blc2020}, which shall be revised in the next Section. The main goal of our work will be, then, to show how to connect different frames by using categorical arguments, leading to the so expected principle of relativity.  
	
	The work is divided as follows. In the next Section, we provide a motivation for our categorical description on relativity. Section \ref{sec.3} is devoted to a brief review of the Galilean and Lorentz groups representation on spacetime (actually, the representation is given in the associated vector space formed by pair of spacetime events. We make an abuse of notation though). The categories named \texorpdfstring{$\mathbcal{Lor}$}{Lor} and  \texorpdfstring{$\mathbcal{Gal}$}{Gal} will be presented in the Section \ref{sec.4}, where we lay our main conclusion - there are no privileged frames, exposed with a sequence of results, namely, both categories are complete and also that inertial frames of reference are small limits. The deep liaison between \texorpdfstring{$\mathbcal{Lor}$}{Lor} and  \texorpdfstring{$\mathbcal{Gal}$}{Gal} is explored in the Section \ref{sec.5}. We show that their connection is a functorial one.
	In Section \ref{sec.6} we present some alternative formulations to our construction and Section \ref{sec.7} is dedicated to the conclusions.

	Throughout the paper, we will use the term spacetime meaning a flat manifold, in both classical and relativistic regimes, without curvature or gravitational effects.
	
	\section{Motivation}
	\label{sec.2}
	
	Our starting point consists of constructing the physical space through the operational philosophy \cite{Opera09}. In order to not run astray, we will not develop this topic in much detail. We could however summarize its core in a somewhat short statement - all the definitions and results must be described in terms of operations, where an operations is an experimental prescription that can be replicated in a laboratory. Although related to Empiricism, Operationalism allows one to extrapolate processes that could be only reproducible in a lab, admitting the usage of abstraction to generate new concepts. The concept of continuity is an standard example. We cannot measure, for instance, $\pi$ as a distance between two points, after all the result is a rational number. We could, nonetheless, measure rational ones, such as, $3; 3,1; 3,14; \hdots$ in such a way that this sequence gets as close to $\pi$ as possible, 
	up to experimental precision. That way, each step is attainable with a clear experimental recipe - an operation, guaranteeing the very existence of the irrational $\pi$ \cite{grb}.
	
	After this brief intro concerning Operationalism, let us review how one can build the notion of a frame of reference, and together with it, the one of space. All the details can be found in \cite{blc2020}. The first observation is that there exists around us bodies whose distance between any two arbitrary and distinct points does not change with time. Once again, all notions such as points, distance and even time can be described by operational prescriptions \cite{grb, rizzuti_is_2021}. For example, a brick is, \textit{per se}, a particular example. We will call them \textit{rigid bodies}. Now, consider a wall. It is itself composed of multiple different bricks. But (we hope) the wall itself is rigid. If we take two bricks in this wall, the distance between them does not change with time, meaning the two of them together are, again, a rigid body.

    We can abstract this idea and glue together an infinite amount of rigid bodies in order to cover the entire universe with points. After this, we ignore the bodies leaving only the abstract idea of the location where the points over the bodies once were. If now we imagine a moving particle, at each moment in time it ``occupies'' a different point. The same can be said for rigid bodies. At this stage, they don't define the points themselves, but are located by them instead.

    To formalize the idea described above, we will define an equivalence relation in the set of all rigid bodies $\mathfrak{K}$. Given two bodies $A, B \in \mathfrak{K}$, we will say they are related if, and only if, the points of $A$ and $B$ do not move with respect to each other. Leveraging in the operationalistic philosophy adopted here, we could judge the relative movement between rigid bodies by looking at vectors defined by pair of points marked 
 in both bodies. Here the notion of vectors we are using comes from an equivalence relation of pairs of points. Consider a set $X$ whose elements we will call \textit{points}. One could think of it as literally marking crosses in rigid bodies with needles and the points being such intersection. Given $P, Q \in X$, consider an oriented line segment starting at $P$ and ending at $Q$, which we denote by $\overrightarrow{(P, Q)}$. Let $S$ denote the set of all such oriented segments by taking pairs of points of $X$. There is an equivalence relation in $S$ known as \textit{parallel transport} - which can be achieved by moving rules supported by squares - denoted by $\top$. Now, take the quotient set $\mathcal{V} = S/\top$. Then, a \textbf{vector} is an element of $\mathcal{V}$. This construction is explored further in \cite{blc2020}.
 
 With this definition, we avoid possible ambiguities in the following sense. If we have circular motion between rigid bodies,  only the distance between pair of points would not avoid relative motion. On the other hand, asking that the vectors are kept fixed, that is, not changing in time,\footnote{Once again, the very notion of time evolution can be operationally defined, see, for example, \cite{rizzuti_is_2021}.} we solve this issue. All the line of thought below is to be understood 
 within this perspective. In symbols:
    \begin{align*}
        A \, \delta \, B \Leftrightarrow \mbox{ the points of $A$ and $B$ do not move with respect to each other.}
    \end{align*}
    To check that it is in fact an equivalence relation, 
	\begin{enumerate}[label=(\alph*)]
		\item $A \, \delta \, A$, after all the points of $A$ do not move with respect to themselves - $A$ is a rigid body itself. 
		
		\item If $A \, \delta \, B$, then the points of $A$ don't move relative to the points of $B$, and, naturally, the points of $B$ don't move relative to the points of $A$. Thus, $B \, \delta \, A$.
		
		\item If $A \, \delta \, B $ and $B \, \delta \, C$, then the points of $A$ don't move relative to the ones in $B$. Also, the points of $B$ don't move relative to the points in $C$. Thus, the points of $A$ don't move relative to the points of $C$. That way, $A \, \delta \, C$.
	\end{enumerate}
	Thus, the $\delta$ defines an equivalence relation, as stated. This result motivates the following.
	\begin{definition} 
		A frame of reference is a class in $\mathfrak{K}/ \delta$. With more details
		$$ [F_1] = \left  \{ F \in \mathfrak{K} | \forall P_1 \in F_1, \forall P \in F, \frac{\partial \overrightarrow{P_1 P}}{\partial t}=\Vec{0} \right \}, $$ 	\end{definition}
    Where $\overrightarrow{P_1 P}$ is the vector defined by the points $P_1$ and $P$. All the points marked in this way located in a particular frame has a central role and deserves the definition below. 
	\begin{definition}
		We call \textit{physical space}, or space for short, the set of all points associated to a particular frame of reference. We denote it by $\mathcal{E}^3$.\footnote{Although this nomenclature has no direct relation to the frame, it is clear from the definition that a space is only defined once the corresponding frame has already been posed.} 
	\end{definition}
	
	We now ask if space is a unique entity. According to our definition, it is not for if a rigid body, say a car, is moving in relation to another, like a tree, they won't be in the same equivalence relation or frame. Another interesting feature, that can be extracted from our previous discussion, is that frames are distinguished from one another (and, of course, the corresponding spaces) by looking to their relative movement. This observation impels one to use Category theory to describe such a relation between frames. To get there, we will first review the basic concepts related to the Galilean and Lorentz groups in the next Section. This discussion will be key when we set up relativity into the categorical formalism.

	\section{Galilean and Lorentz group representation - a revision}
	\label{sec.3}    
	
	Our central idea when defining the representations of Galilean and Lorentz groups is to find transformations between frames, moving in relation to one another. The discussion we have conducted so far suggests that rotations and translations can be left out of our considerations once these two operations keep the frame inert.
	
	Firstly, in the same fashion the notion of space was built, we now do it for what is going to be called \textit{spacetime}. It is a set of \textit{events}, which in turn, are occurrences taking place in a small enough region of space that can be approximated by a point and happens so fast that its duration may be approximated by a instant. We denote the spacetime by $\mathcal{ST}$. 
	
	The following definition is also central.  
	\begin{definition}
		An \textit{observer} is a coordinate system in $\mathcal{ST}$.
	\end{definition}
	
	\begin{remark}
		Usually, the very definition of an observer is done through future-pointing timelike curves in the spacetime \cite{sachs2012general}. However, the differentiable structure of the spacetime will not be needed here. We only have to require the assumption - commonly accepted by the community - that spacetime is a manifold \cite{kopczynski1992spacetime}. This way, the definition of an observer is supposed to be interpreted as a choice of basis on the vector field obtained from pair of events. 
	\end{remark}
	
	The above remark mentions the concept of a vector space associated to the spacetime itself. One can in fact build the vector space using pair of events, with the structure of an affine space \cite{kopczynski1992spacetime}. The vector space will be used to define the representations of Galilean and Lorentz groups \cite{adb}. Let us discuss it in more details. Given a coordinate system $(T, \vec{X})$, the aforementioned vector space is given by elements of the form 
	\begin{equation}\label{3.1}
		(t, \vec{x}) \coloneqq (T(e_1) - T(e_0), \vec{x}(e_1) -\vec{x}(e_0)),
	\end{equation}
	where $e_0, e_1 \in \mathcal{ST}$. The synthetic notation $(t, \vec{x})$ is used to represent the difference of coordinates (temporal and spatial) of the pair of events. The sum and multiplication by numbers are defined in the usual way. Some observations are in order. 
	\begin{enumerate}[label=(\roman*)]
		\item It is usual to find in the literature the term (flat) ``spacetime'' for the vector space constructed above \cite{oneill}. To be more precise, there is a difference between the Minkowski spacetime (isometric to $(\mathds{R}^4, \eta)$, where $\eta$ is the standard pseudo-metric) to the Newtoninan one $\mathds{R}\times \mathds{R}^3$, see \cite{kopczynski1992spacetime, oneill} for details. This distinction will not be relevant to our discussion though.
		
		\item Throughout the text, we will abuse the notation to say that Galilean and Lorentz transformations connect frames of reference. We must not forget that these transformations are applied on the vector space constructed in \eqref{3.1}, from the fixed coordinates in $\mathcal{ST}$.		
	\end{enumerate}
	
	We still need one last concept, that of an inertial frame of reference. Here, we give a somewhat heuristic notion and systematically operational. Let us consider particles sufficiently far apart of other bodies, such that they do not interact with anything around them. We expect that these particles, which we will call \textit{free}, follow the simplest possible trajectory in a frame of reference - a point or a straight line. In physics, there exists an special interest in this case, as we know that the laws of classical mechanics are invariant with respect to changes between these frames of reference. With that, we have the following.
	\begin{definition}
		A frame of reference is called \textbf{inertial} when the trajectory of any free particle is a point or a straight line.
	\end{definition}
	
	Above all we are interested in frames of reference that move in relation to others, after all relative motion is exactly what differentiates them. We will always compare two frames of reference. Therefore, we need to assume some premises as a starting point.
	
	In the following premises, let $A, B, C$ be inertial frames of reference.
	
	\begin{premise}\label{pr1}
		There exist frames of reference.
	\end{premise}
	
	\begin{premise}\label{pr2}
		Given any two frames of reference $A, B$, there exists a unique velocity $\vec{v}$ such that $B$ moves with velocity $\vec{v}$ in relation to $A$.
	\end{premise}
	
	\begin{premise}\label{pr3}
		If $B$ moves with velocity $\vec{v}$ in relation to $A$, then we can affirm that $A$ moves with velocity $-\vec{v}$ in relation to $B$.
	\end{premise}
	
	\begin{premise}\label{pr4}
		If $B$ moves with velocity $\vec{v}$ in relation to $A$ and $C$ moves with velocity $\vec{u}$ in relation to $B$, then $C$ moves with velocity $\vec{u} + \vec{v}$ in relation to $A$. Actually, we have here two premises once the sum may represent both the classical and relativistic addition of velocities. Moreover, in the relativistic case, when the velocities are non-parallel, the composition involves a Thomas rotation. Yet, the angle associated with the rotation is defined by both $\Vec{v}$ and $\Vec{u}$. See the comment after Definition \ref{defGalLortrans}.
	\end{premise}
	
	In relation to the \textbf{Premise \ref{pr3}}, we need to emphasize that there is a difference between a frame of reference, let us say $C$, that moves with velocity $-\vec{v}$ in relation to $A$ and say that $A$ moves with velocity $-\vec{v}$ in relation to $B$. In this case, $C$ moves with velocity $(-\vec{v})+ (- \vec{v})$ in relation to $B$.
	
	\begin{definition}\label{defGalLortrans}
		Let $A$ be a frame of reference and $B$ be a frame of reference moving with velocity $\vec{v}$ in relation to $A$.
		\begin{enumerate}[label=(\roman*)]
			\item A \textbf{Galilean transformation} is a function $\Gamma_{\vec{v}}: A \rightarrow B$, defined by
			\begin{eqnarray*}
				\Gamma_{\vec{v}}(t,\vec{x}) = (t, \vec{x} - \vec{v}t).
			\end{eqnarray*}

            Formally, an arbitrary $\Gamma$ can also involve a rotation, which is not relevant for our discussion in the sense that the sectors of boosts and rotations do not mix with each other when one composes Galilean transformations.
			\item A \textbf{Lorentz transformation} is a function $\Lambda: A \rightarrow B$, which may be composed of rotations and boosts. For the particular case of a boost in the direction defined by  $\Vec{v}$ - the velocity connecting $A$ and $B$ - it reads 
			\begin{eqnarray*}
				\Lambda_{\vec{v}}(t, \vec{x}) = \left(\gamma_{\vec{v}}\left (t - \dfrac{\langle \vec{v}, \vec{x} \rangle}{c^2}\right ), \vec{x} + (\gamma_{\vec{v}} - 1) \frac{\langle\vec{v},\vec{x}\rangle}{v^2} \vec{v} - \gamma_{\vec{v}}\vec{v}t\right).
			\end{eqnarray*}
		\end{enumerate}
	\end{definition}
	In the equation above, $\langle \cdot, \cdot \rangle$ represents the scalar product in $\mathbb{R}^3$. Contrary to the first case above, two successive boosts may contain a rotation. We will return to this intriguing point in a while. For now, in the relativistic case, we have to notice the following. If $B$ moves with respect to $A$ with velocity $\Vec{v}$ and $C$ moves with respect to $B$ with velocity $\Vec{u}$, the connection between $A$ and $C$ involves not only a boost, but also a rotation. We point out that, although there is a rotation involved, the corresponding angle is uniquely determined by the relativistic composition of velocities. In this case the transformation from $A$ to $C$ is given by the expression
        \begin{align*}
            \Lambda = \Lambda_{\mathbcal{R}} \circ \Lambda_{\vec{u}\, \oplus_{\mathcal{S}} \, \vec{v}},
        \end{align*}
        where $\mathcal{R}$ is known as the Thomas rotation, described in full details in \cite{ungar1, ungar2, sexl}.

	\begin{remark}
		$\gamma_{\vec{v}}$ is a constant given by
		\begin{eqnarray*}
			\gamma_{\vec{v}} = \left( \sqrt{1 - \dfrac{|\vec{v}|^2}{c^2}} \right)^{-1}
		\end{eqnarray*}
		known as \textbf{Lorentz factor}. In the equation above $c$ is the speed of light.
	\end{remark}
	
	When dealing with Lorentz transformations, we are working in a relativistic context. We will not go in detail about the origins and properties of physical systems in this context, with exception of the fact that velocities such that $|\vec{v}| > c$ are not permitted. Therefore, we need to separate Premise \ref{pr2} in two cases.

	\begin{premise*3.1}[Classic]
		Given any two frames of reference $A, B$, there is a unique velocity $\vec{v}$, where $|\vec{v}| \in [0, +\infty)$, such that $B$ moves with velocity $\vec{v}$ in relation to $A$ and a unique Galilean transformation associated with $\vec{v}$.
	\end{premise*3.1}
	
	\begin{premise*3.2}[Relativistic]
		\label{premise3.2.2}
		Given any two inertial frames of reference $A, B$, there is a unique velocity $\vec{v}$, where $|\vec{v}| \in [0, c)$, such that $B$ moves with velocity $\vec{v}$ in relation to $A$ and a unique Lorentz transformation associated with $\vec{v}$. There are instances where this velocity may define a rotation, as explained above.
	\end{premise*3.2}
	
	We still need an observation. As in relativistic contexts the allowed speeds are limited by $c$, we cannot use the same sum of velocities as in the classical context. Usually, it is calculated in terms of the associated Lie Algebra, that we will omit here, as we are interested in the group structure as well as aiming a bird's eye view of relativity via category theory. A thorough discussion, with a closed formula for the addition of velocities in the relativistic regime can be seen in \cite{sexl} though. To emphasize this difference, some authors use different symbols for the relativistic and classical velocity sums. We will follow this prescription, to avoid confusion. The classical addition of velocities will be denoted by $\oplus_{\mathcal{C}}$ as well as $\oplus_{\mathcal{S}}$ stands to the (special) relativistic instance. Finally, we observe that Premise \ref{pr4} actually represents two distinct axioms, one for the relativistic and the other to the classical case.
	
	\begin{prop}\label{properties}
		With these definitions, the Galilean transformations $(\Gamma)$ have the following properties: 
		\begin{enumerate}[label = ($\Gamma.$\roman*)]
			\item linearity;
			
			\item $\Gamma_{\vec{0}} = \mathrm{id}_{A}$;
			
			\item $\Gamma_{\vec{u}} \circ \Gamma_{\vec{v}} = \Gamma_{\vec{u}\, \oplus_{\mathcal{C}} \, \vec{v}}$;
			
			\item $(\Gamma_{\vec{v}})^{-1} = \Gamma_{-\vec{v}}$.			
		\end{enumerate}
        \end{prop}
        \begin{proof}
            The proof for these results are straightforward and can be found in many places such as \cite{adb}.
        \end{proof}

        The relativistic case is more intricate. As noted above, the composition of two Lorentz boosts in arbitrary directions can be decomposed in a rotation, together with the boost parametrized by the special composition of velocities. 

        \begin{prop}
		The Lorentz transformations $(\Lambda)$ obey the similar properties as above, 
		\begin{enumerate}[label = ($\Lambda.$\roman*)]
			\item linearity;
			
			\item $\Lambda_{\vec{0}} = \mathrm{id}_{A}$;
			
			\item $\Lambda_{\vec{u}} \circ \Lambda_{\vec{v}} = \Lambda_{\mathbcal{R}} \circ \Lambda_{\vec{u}\, \oplus_{\mathcal{S}} \, \vec{v}}$;
			
			\item $(\Lambda_{\vec{v}})^{-1} = \Lambda_{-\vec{v}}$.			
		\end{enumerate}
	\end{prop}
 
        \begin{proof}
            Items $(i)$ and $(ii)$ of this proposition have the same proof as the classical case. However, the other two need a more in depth explanation. As stated above, the composition of velocities in the relativistic case needs a ``correction'', given by what is known as a Thomas rotation. As this calculation is extensive, we will omit it here, but we refer the interested reader to \cite{ungar1, ungar2}, proving item \textit{(iii)}.

            Now, we proceed to item $(iv)$. One needs to show that the composition of $\Lambda_{\vec{v}}$ with $\Lambda_{-\vec{v}}$ is the identity. But, by item $(iii)$ we have
            \begin{align*}
                \Lambda_{\vec{v}} \circ \Lambda_{-\vec{v}} = \Lambda_{\mathcal{R}} \circ \Lambda_{\vec{v} \oplus_{\mathcal{S}} (-\vec{v})}.
            \end{align*}
            As the relativistic sum $\vec{v} \oplus_{\mathcal{S}} (-\vec{v}) = \Vec{0}$, item $(i)$ guarantees that the last boost is the identity. It is also simple to check that the Thomas rotation for $\Vec{v}$ and $-\vec{v}$ is the identity, proving that $\Lambda_{-\vec{v}} = (\Lambda_{\Vec{v}})^{-1}$,as desired.
        \end{proof}
	
	Note that, property $(ii)$ means that we can interpret $A$ moving with velocity $\vec{0}$ in relation to $A$, which is consistent with premise \ref{pr2}. Also, properties $(iii)$ and $(iv)$ are consistent with premises \ref{pr3} and \ref{pr4}.
	
	As a consequence of item $(iv)$ of Proposition \ref{properties}, we have the following.
	
	\begin{corolario}
		Lorentz and Galilean transformations are linear isomorphisms.
	\end{corolario}
	
	It is important to note that Premise \ref{pr2} means that between both inertial frames of reference, there is a unique Lorentz or Galilean transformation between them.
	
	As a last result, we have.
	
	\begin{prop}\label{propuni}
		Given two inertial frames of reference $A$ and $B$, there is a unique Galilean (or Lorentz) transformation between $A$ and $B$.
	\end{prop}
	\begin{proof}[\unskip\nopunct]
		\textbf{Proof:}
		By \textbf{Premise} \ref{pr2}, we know there is a unique velocity $\vec{v}$ (with $|\vec{v}| \in [0, c)$ if necessary), such that $B$ moves with velocity $\vec{v}$ in relation to $A$. Therefore, there exist transformations
		\begin{eqnarray*}
			\Gamma_{\vec{v}}: A \rightarrow B \; \mbox{ and } \; \Lambda_{\vec{v}}: A \rightarrow B.
		\end{eqnarray*}
		As each velocity is unique, we also have the uniqueness of the transformations.
	\end{proof}
	
	\section{Categories \texorpdfstring{$\mathbcal{Lor}$}{Lor} and \texorpdfstring{$\mathbcal{Gal}$}{Gal}}
	\label{sec.4}
	
	We now can make use of the framework above in the categorical formalism. We begin with the definitions below.
	
	\begin{definition}
		We define the category $\mathbcal{Gal}$ as the category which objects are inertial frames of reference and which morphisms are Galilean transformations.
	\end{definition}
	
	\begin{definition}
		We define the category $\mathbcal{Lor}$ as the category which objects are inertial frames of reference and which morphisms are Lorentz transformations.
	\end{definition}
	
	Note that, $\mathbcal{Gal}$ and $\mathbcal{Lor}$ are in fact categories by Proposition \ref{properties}, that is, we know that for each inertial frame of reference, there exists an identity (namely, $\Gamma_{\vec{0}}$ and $\Lambda_{\vec{0}}$) and, for frames of reference, $A, B, C$ and morphisms $\Gamma_{\vec{v}}, \Lambda_{\vec{v}}: A \rightarrow B$ and $\Gamma_{\vec{u}}, \Lambda_{\vec{u}}: B \rightarrow C$, there are compositions
	\begin{eqnarray*}
		\Gamma_{\vec{u}} \circ \Gamma_{\vec{v}} = \Gamma_{\vec{u} \, \oplus_{\mathcal{C}} \, \vec{v}}: A \rightarrow C
	\end{eqnarray*}
	and
	\begin{eqnarray*}
		\Lambda_{\vec{u}} \circ \Lambda_{\vec{v}} = \Lambda_{\mathbcal{R}}\Lambda_{\vec{u} \, \oplus_{\mathcal{S}} \, \vec{v}}: A \rightarrow C.
	\end{eqnarray*}
	
	Note also that by Proposition \ref{propuni}, each inertial frame of reference is a zero object in the categories $\mathbcal{Gal}$ and $\mathbcal{Lor}$, that is, are initial and terminal objects. We can also show a even stronger result.
	
	\begin{prop}
		$\mathbcal{Gal}$ and $\mathbcal{Lor}$ are complete categories.
	\end{prop}
	\begin{proof}[\unskip\nopunct]
		\textbf{Proof:}
		Let $\boldsymbol{I}$ be a small category and $D: \boldsymbol{I} \rightarrow \mathbcal{Gal}$ a diagram. Given any object $A \in \mathbcal{Gal}$, for all $I,J \in \boldsymbol{I}$ we have the commutative diagram
		\begin{eqnarray*}
			\begin{tikzcd}
				& & D(I) \arrow[]{dd}{D(f)}
				\\
				A \arrow[]{urr}{\Gamma_{\vec{v}}} \arrow[swap]{drr}{\Gamma_{\vec{u}}} & &
				\\
				& & D(J)
			\end{tikzcd}
		\end{eqnarray*}
		where $D(I)$ moves with velocity $\vec{v}$ in relation to $A$ and $D(J)$ with velocity $\vec{u}$ in relation to $A$, by the uniqueness of Galilean transformations. Then, every frame of reference defines a cone in the category. Moreover, if we have another cone, $(B \stackrel{\Gamma_{\vec{v}_I}}{\rightarrow} D(I))_{I \in \boldsymbol{I}}$ in the category, we know, by Proposition \ref{propuni}, that there exists a unique Galilean transformation $\Gamma_{\vec{w}}: B \rightarrow A$. Thus, the diagram
		\begin{eqnarray*}
			\begin{tikzcd}
				& & & D(I) \arrow[]{dd}{D(f)}
				\\
				B \arrow[dashed, pos=0.8]{rr}{\Gamma_{\vec{w}}} \arrow[]{urrr}{\Gamma_{\vec{v}_I}} \arrow[swap]{drrr}{\Gamma_{\vec{u}_J}} & & A \arrow[swap]{ur}{\Gamma_{\vec{v}}} \arrow[]{dr}{\Gamma_{\vec{u}}} &
				\\
				& & & D(J)
			\end{tikzcd}
		\end{eqnarray*}
		commutes, for every $I,J \in \boldsymbol{I}$ and $f: I \rightarrow J$, again by Propostion \ref{propuni}. Hence $A$ is a limit of the diagram $D: \boldsymbol{I} \rightarrow \mathbcal{Gal}$. As $\boldsymbol{I}$ is arbitrary, we conclude that $\mathbcal{Gal}$ is complete. In other words, $\mathbcal{Gal}$ has limits of shape $\boldsymbol{I}$, for every small category $\boldsymbol{I}$. Analogously, it is possible to show that $\mathbcal{Lor}$ is complete. For that, the Thomas rotation must be taken into account. In particular, changing $\Gamma \leftrightarrow \Lambda$ in the first diagram above, the commutativity is guaranteed by the uniqueness of Lorentz transformation as in \textbf{Premise \ref{premise3.2.2}}, namely both $D(f)$ and $\Lambda_{\mathbcal{R}} \Lambda_{\vec{u} \, \oplus_{\mathcal{S}} (-\Vec{v})}$  are maps from $D(I)$ to $D(J)$ and therefore they coincide assuring that the diagram commutes. In the second one, we also change $\Gamma \leftrightarrow \Lambda$. To show it commutes, we note that $\Lambda_{\Vec{w}}$, $\Lambda_{\mathbcal{R}} \Lambda_{(-\vec{v}) \, \oplus_{\mathcal{S}} (\Vec{v}_I)}$ and $\Lambda_{\mathbcal{R}} \Lambda_{(-\vec{u}) \, \oplus_{\mathcal{S}} (\Vec{u}_J)}$ are maps from $B$ to $A$. Then uniqueness guarantees they are all the same, assuring that the diagram commutes.
	\end{proof}
	
	\begin{corolario}
		Any inertial frame of reference is a small limit in the categories $\mathbcal{Gal}$ and $\mathbcal{Lor}$.
	\end{corolario}
	\begin{proof}[\unskip\nopunct]
		\textbf{Proof:}
		Following the proof of the previous Proposition, any object in $\mathbcal{Gal}$ and $\mathbcal{Lor}$ are limits of the diagram, as the role of the vertices of any cone is interchangeable.
	\end{proof}
	
	The most interesting fact in this corollary is that it implies something expected, that is, there are no privileged inertial frames of reference. If given a diagram there existed some objects that are limits and some that are not, that would mean there is a difference in how some frames of reference perceive the others. We should remember, as mentioned before, that in order to say something about inertial frames of reference, we must always compare them, that is, if we consider frames of reference $A$, $B$ and $C$, there cannot be any difference in our description of $B$ and $C$ in relation to $A$ or of $B$ and $C$ in relation to $B$. Therefore, the corollary above seems to be consistent with the principle of relativity, that says the laws of physics must be the same in any inertial frame of reference. In fact, if there is no privileged frame, one may use anyone of them to describe the mathematical laws governing physics processes. 
	
	Another remarkable point that should be considered is that groups are not complete categories. Indeed groups do not have initial or terminal objects (with exception of the trivial group). The reason for this is that groups usually have more than one map $G \rightarrow G$, where $G$ is the single object in the group seen as a category. However, with our approach we obtain a complete category.
	
	Note that the Lorentz and Galilean groups are in fact groups in the algebraic sense, which could have been defined as categories $\mathbcal{Lor}$ and $\mathbcal{Gal}$ with a single object. The reason we haven't defined our categories this way is that we believe each inertial frame of reference gives birth to a distinct space. One way of seeing this is the fact that the set of accessible events for each frame of reference is different. Moreover, the construction of the physical space we did in previous works seems to also point in this direction \cite{blc2020, ldb2019}.
	
	\section{The Limit Functor}
	\label{sec.5}
	
	Because of the similarities between the categories $\mathbcal{Gal}$ and $\mathbcal{Lor}$, it is natural to ask if there is a way to connect them. In order to do so, we take inspiration from a phenomenon well known to physics, namely, that in the limit $c \to +\infty$, we recover classical formulae from the relativistic ones. 
	
	More precisely, given a velocity $\vec{v}$ with $|\vec{v}| < c$, the Lorentz factor is given by
	\begin{equation}
		\gamma_v = \dfrac{1}{\sqrt{1 - \dfrac{|\vec{v}|^2}{c^2}}}.
	\end{equation}
	Hence, applying the limit $c \to +\infty$, we obtain $\gamma_{\vec{v}} \to 1$.
	
	As Lorentz transformations are linear, in particular, they are continuous, that is, 
	\begin{align}
		\lim\limits_{c \to +\infty} \Lambda_{\vec{v}}(t, \vec{x}) &= \lim\limits_{c \to +\infty}\left(\gamma_{\vec{v}}\left (t - \dfrac{\langle \vec{v}, \vec{x} \rangle}{c^2}\right ), \vec{x} + (\gamma_{\vec{v}} - 1) \frac{\langle\vec{v},\vec{x}\rangle}{v^2} \vec{v} - \gamma_{\vec{v}}\vec{v}t\right) \nonumber \\ & = \Gamma_{\vec{v}}(t, \vec{x}),
	\end{align}
	where we are considering the particular case of a boost. However, both the relativistic addition of velocities and the Thomas rotation are continuously deformed into the classical addition of velocities as well as the identity mapping in the $c \to +\infty$ limit \cite{ungar1}. From this, we propose the following.
	
	\begin{definition}
		The \textbf{limit functor} $L: \mathbcal{Lor} \rightarrow \mathbcal{Gal}$ defined by:
		\begin{enumerate}[label=(\roman*)]
			\item for every object $A \in \mathbcal{Lor}$ $L(A) = A$;
			
			\item for each pair of objects $A,B \in \mathbcal{Lor}$ and morphisms $\Lambda_{\vec{v}}: A \rightarrow B$, $L(\Lambda_{\vec{v}}) = \Gamma_{\vec{v}}$.
		\end{enumerate}
	\end{definition}
	
	To show that $L$ defines indeed a functor, we need to show that compositions and identities are preserved by $L$. Thus, let us consider $A, B, C \in \mathbcal{Lor}$, such that $B$ moves with velocity $\vec{v}$ in relation to $A$ and $C$ with velocity $\vec{u}$ in relation to $B$ and $\Lambda_{\vec{v}}: A \rightarrow B$, $\Lambda_{\vec{u}}: B \rightarrow C$. Then, by Proposition \ref{properties}, we have
	\begin{eqnarray*}
		L(\Lambda_{\vec{u}} \circ \Lambda_{\vec{v}}) = L( \Lambda_{\mathbcal{R}} \Lambda_{\vec{u} \, \oplus_{\mathcal{S}} \, \vec{v}}) = \Gamma_{\vec{u} \, \oplus_{\mathcal{C}} \, \vec{v}} = \Gamma_{\vec{u}} \circ \Gamma_{\vec{v}} = L(\Lambda_{\vec{u}}) \circ L(\Lambda_{\vec{v}}).
	\end{eqnarray*}
	Now, let $A \in \mathbcal{Lor}$. Then,
	\begin{eqnarray*}
		L(\mathrm{id}_A) = L(\Lambda_{\vec{0}}) = \Gamma_{\vec{0}} = \mathrm{id}_{A} = \mathrm{id}_{L(A)}.
	\end{eqnarray*}
	Therefore, $L$ defines indeed a functor.
	
	Because of the way we defined the categories $\mathbcal{Lor}$ and $\mathbcal{Gal}$, we have our next result.
	\begin{prop}
		$L$ is full and faithful.
	\end{prop}
	\begin{proof}[\unskip\nopunct]
		\textbf{Proof:}
		Let $A, B$ be inertial frames of reference, such that $B$ moves with velocity $\vec{v}$ in relation to $A$. Thus, as $\mathbcal{Lor}(A, B) = \{\Lambda_{\vec{v}}\}$ and $\mathbcal{Gal}(L(A), L(B)) = \{\Gamma_{\vec{v}}\}$, the map $\mathbcal{Lor} \rightarrow \mathbcal{Gal}$ is injective and surjective. Therefore, $L$ is full and faithful.
	\end{proof}
	
	Moreover, we can show the proposition below. 
	\begin{prop}
		L preserves limits.
	\end{prop}
	\begin{proof}[\unskip\nopunct]
		\textbf{Proof:}
		Given a small category $\boldsymbol{I}$ and a diagram $D: \boldsymbol{I} \rightarrow \mathbcal{Lor}$, we have already seen that any object in $\mathbcal{Lor}$ is a limit of this diagram. Still, as any object in $\mathbcal{Gal}$ is a limit of the diagram $L \circ D$, in particular, the image of a limit in $\mathbcal{Lor}$ through $L$ is a limit in $\mathbcal{Gal}$. Therefore, $L$ preserves limits.
	\end{proof}
	
	Now, we can use the adjoint functor theorem \cite{leinster_basic_2014} to find a funtor $\mathbcal{Gal} \rightarrow \mathbcal{Lor}$.
	
	\begin{prop}
		$L$ has a left adjoint.
	\end{prop}
	\begin{proof}[\unskip\nopunct]
		\textbf{Proof:}
		Notice that given $A,B \in \mathbcal{Lor}$, such that $B$ moves with velocity $\vec{v}$ in relation to $A$, we have the set $\mathbcal{Lor}(A, B) = \{\Lambda_{\vec{v}}\}$. Hence, $\mathbcal{Lor}$ is a small and complete category.
		
		Now, let $A \in \mathbcal{Gal}$ and consider a category $(A \Rightarrow L)$. Let $B, B' \in \mathbcal{Lor}$ such that $L(B)$ moves with velocity $\vec{v}$ in relation to $A$ and $L(B')$ with velocity $\vec{u}$ in relation to $A$. Thus, $(B, \Gamma_{\vec{v}})$ and $(B', \Gamma_{\vec{u}})$ are objects of $(A \Rightarrow L)$. Moreover, by Proposition \ref{propuni}, there exists a unique Lorentz transformation $\Lambda: B \rightarrow B'$ (which may be composed by a Thomas rotation, irrelevant as it tends to the identity with the functor) and, by the uniqueness of the maps in $\mathbcal{Gal}$, the diagram
		\begin{eqnarray*}
			\begin{tikzcd}
				A \arrow[swap]{dr}{\Gamma_{\vec{u}}} \arrow[]{r}{\Gamma_{\vec{v}}} & B \arrow[]{d}{L(\Lambda)}
				\\
				& B'
			\end{tikzcd}
		\end{eqnarray*}
		commutes, i.e., $\Lambda_{\vec{w}}: (B, \Gamma_{\vec{v}}) \rightarrow (B', \Gamma_{\vec{u}})$ is a map in $(A \Rightarrow L)$ and it is unique, again by Proposition \ref{propuni}. Then, $(B, \Gamma_{\vec{v}})$ is initial in $(A \Rightarrow L)$ and, in particular, there exists a weakly initial set in $(A \Rightarrow L)$ for every $A \in \mathbcal{Gal}$. Therefore, by the general adjoint functor theorem, $L$ admits a left adjoint, $M: \mathbcal{Gal} \rightarrow \mathbcal{Lor}$.
	\end{proof}
	
	Hence, it is possible to define a map that takes us from the classical context to a relativistic one, which is not surprising, given the extremely similar structure of the categories $\mathbcal{Lor}$ and $\mathbcal{Gal}$.
	
	\section{Alternative formulations}
	\label{sec.6}
	
	The categories defined so far are just part of the Lorentz and Galilean groups. In fact, these groups still contain rotations and translations (in the case of Poincaré group). To add these transformations, we may define a category which objects are ordered pairs $(A, \mathcal{O})$, where $A$ is an inertial frame of reference and $\mathcal{O}$ is a point in $\mathcal{E}^3$ called the origin. By doing this, we can add rotations and translations in $\mathbcal{Lor}$ and $\mathbcal{Gal}$ so that rotations are automorphisms and translations act on the second entry of the ordered pair, changing the origin. In this way, these categories can be represented by the diagram
	
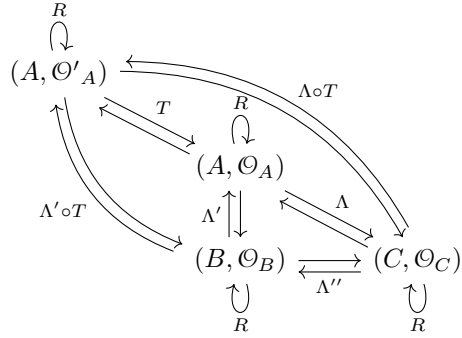
\begin{equation*}
		\begin{tikzcd}
			(A,\mathcal{O'}_A) \arrow[bend left]{ddrr}{} \arrow[bend right]{ddr}{} \arrow{dr}{T} \arrow[loop above]{l}{R}
			\\
			& (A, \mathcal{O}_A) \arrow{d}{} \arrow{dr}{\Lambda} \arrow[shift left=1ex]{ul}{} \arrow[loop above]{l}{R}
			\\
			& (B, \mathcal{O}_B) \arrow[]{r}{} \arrow[bend left, shift left=1ex]{uul}{\Lambda' \circ T} \arrow[shift left=1ex]{u}{\Lambda'} \arrow[loop below]{l}{R} 
			& (C,\mathcal{O}_C) \arrow[bend right, shift right=1ex,swap]{uull}{\Lambda \circ T} \arrow[shift left=1ex]{ul}{} \arrow[shift left=1ex]{l}{\Lambda''} \arrow[loop below]{l}{R}
		\end{tikzcd}
	\end{equation*}
	where $T$ are translations, $R$ rotations and $\lambda$ Lorentz (or Galilean) boosts.
	
	By using this definition, we lose the uniqueness of the maps between different objects as $\Lambda \circ \mathrm{id}_A \neq \Lambda \circ R$, for some rotation $R$ and some Lorentz (or Galilean) pure boost $\Lambda$. Note, however, that the category we defined before can be seen as a subcategory of this last one, giving us a partial result.
	
	Another way of adding these transformations would be by defining the objects as triples $(A, \mathcal{O}, \mathcal{B})$, where $A$ is an inertial frame of reference, $\mathcal{O} \in \mathcal{E}^3$ is an origin and $\mathcal{B}$ is a basis of $\mathcal{ST}$. By doing this, we can maintain the uniqueness of maps between distinct objects.
	
	The reason we did not use this last definition is the loss of the interpretation that the objects are entities (or, operationally, people). We mean that, to each entity in the universe, we know they are in a certain point origin and they are associated to an inertial frame of reference (we are ignoring the non-inertial case). However, the choice of basis is arbitrary: this supposed entity can choose any basis of $\mathcal{E}^3$ to describe their associated space.
	
	There still exists an interesting interpretation to define a category using pairs. With this interpretation, where objects are pairs $(A, \mathcal{O})$ of an inertial frame of reference and a origin, we know that the spaces defined by them are intrinsically distinct, that is, if an entity moves in relation to another, the events of spacetime experienced by each one of them are different and the same occurs for the entities in distinct origins of the Universe. Thus, the objects are different for a more fundamental reason than a simple arbitrary choice of basis.
	
	An important reason for us to have a more profound justification to say that the objects are different is the fact that our structure is still equivalent to that of a group. As observed previously, groups seem as categories that do not posses (in general) initial objects, i.e., for any group, we can define anything as an object in the category, ``separating'' the maps of the group, obtaining the structure defined here. However, for any group, this definition would be completely arbitrary, while in the case of Lorentz and Galilean groups, there seems to exist a fundamental way of defining distinct objects in the category. Interestingly, this also appears to be consistent with the principle of relativity, as there are no privileged frames.

	\section{Conclusions}
	\label{sec.7}
	
	In this work we have presented a non-trivial example of category theory applied to physics. We make use of its unifying power to connect frames of reference in the relativity context. It seems to be useful once group theory does not differentiate transformations of coordinates in the same frame (rotations / translations) or in different frames (boosts).
	
	Leveraging on a fully operational prescription, different inertial frames of reference are only differentiated by their relative movement. This restriction allows us to construct the $\mathbcal{Gal}$ and $\mathbcal{Lor}$ categories, whose objects are frames of reference and morphisms are clearly boosts connecting them. The technical requirement to call $\mathbcal{Gal}$ and $\mathbcal{Lor}$ categories are in fact obeyed: the composition of boosts is again a boost (up to a Thomas rotation in the relativistic case) and a boost parametrized by the velocity $\vec{v} = \vec{0}$ is the identity. One of our central results consists of showing that two categories are complete, that is, for every small category $\boldsymbol{I}$, $\mathbcal{Gal}$ and $\mathbcal{Lor}$ have limits of shape $\boldsymbol{I}$. As a corollary, any inertial frame is a small limit in the categories $\mathbcal{Gal}$ and $\mathbcal{Lor}$. Translating it to the physics language, we are concluding that there are no privileged frames. We have also provided alternative formulations to take into account rotations and translations and the result of having no special frame, or even direction and origin, still holds. Thus, we may infer, by pure categorical terms, the principle of relativity. In effect, a particular law of physics is valid in a frame, if and only if, it also holds true in any other, due to their equivalence so exposed. 
	
	Parallel to the mean result stated above, and as a final comment, we have also shown that the connection between $\mathbcal{Gal}$ and $\mathbcal{Lor}$ is a functorial one. The Transition from $\mathbcal{Lor}$ to $\mathbcal{Gal}$ is directly obtained by imposing the analytical limit $c \rightarrow +\infty$. On the other hand, it is also guaranteed by the general adjoint functor theorem the return from $\mathbcal{Gal}$ to $\mathbcal{Lor}$, although the theorem does not specify the correct inversion from classical to the relativistic regime. 
	
	\section*{Acknowledgments}
	This work was supported by Programa Institucional de Bolsas de Iniciação Científica - XXXIII BIC/Universidade Federal de Juiz de Fora - 2020/2021, project number 47868.
	
	The authors are in debt with C. Duarte. He firstly presented us to both Category and Resource theories and is permanently encouraging us to apply new abstract mathematical formalism to description of the physical reality.

\end{document}